\documentclass[11pt,a4paper,reqno]{article}
\usepackage{amsmath}
\usepackage{graphicx}
\usepackage{amsfonts}
\usepackage{amssymb}
\usepackage{MnSymbol}
\usepackage{amsthm}
\usepackage[left=3cm, right=3cm, top=3cm, bottom=3.5cm]{geometry}
\usepackage{indentfirst}
\usepackage[all]{xy}
\usepackage{appendix}
\usepackage[colorlinks=true,linkcolor=blue]{hyperref}
\usepackage{mathrsfs} %%pour le style calligraphie en mathscr
\usepackage{ytableau}
\usepackage{tikz-cd}
\usepackage{enumitem}
\setitemize{noitemsep,topsep=0pt,parsep=0pt,
partopsep=0pt,itemindent=12pt,leftmargin=0pt}


\newtheorem{theorem}{Theorem}[section]
\newtheorem{lemma}[theorem]{Lemma}
\newtheorem{proposition}[theorem]{Proposition}

\newtheorem{definition}[theorem]{Definition}

\newcommand{\build}[3]{\mathrel{\mathop{\kern 0pt#1}\limits_{#2}^{#3}}}

\def\SU{{\mathrm{SU}}}
\def\su{{\mathfrak{su}}}
\def\U{{\mathrm U}}
\def\Z{{\mathbb Z}}
\def\N{{\mathbb N}}

\def\C{\mathbb{C}}
\def\Im{\mathrm{Im}}

\def\geq{\geqslant}
\def\leq{\leqslant}

\providecommand{\keywords}[1]
{
  \small
  \textbf{\textit{Keywords---}} #1
  \normalsize
}

\title{Large $N$ behaviour of the two-dimensional \\ Yang--Mills partition function}
\author{Thibaut Lemoine\thanks{Institut de Recherche Math\'ematique Avanc\'ee, UMR 7501, Universit\'e de Strasbourg et CNRS, 7 rue Ren\'e Descartes, 67000 Strasbourg, France. E-mail: \texttt{thibaut.lemoine@unistra.fr}}}

\begin{document}
\maketitle

\begin{abstract}
We compute the large $N$ limit of the partition function of the Euclidean Yang--Mills measure on  orientable compact surfaces with genus $g\geq 1$ and non-orientable compact surfaces with genus $g\geq 2$, with structure group the unitary group $\U(N)$ or special unitary group $\SU(N)$. Our proofs are based on asymptotic representation theory: more specifically, we control the dimension and Casimir number of irreducible representations of $\U(N)$ and $\SU(N)$ when $N$ tends to infinity. Our main technical tool, involving `almost flat' Young diagram, makes rigorous the arguments used by Gross and Taylor \cite{GT} in the setting of QCD, and in some cases we recover formulae given by Douglas \cite{Dou} and Rusakov \cite{Rus}. 
\end{abstract}


\tableofcontents

\section{Introduction}

In his seminal paper \cite{Hoo}, 't Hooft discovered that $\SU(N)$ and $\U(N)$ two-dimensional gauge theories become easier to understand when considering the limit $N\to\infty$, thanks to combinatorial simplifications. After that, the idea of studying large $N$ limits of matrix models flourished, in particular in the case of Quantum Chromodynamics in two dimensions, or $\mathrm{QCD}_2$ \cite{DK,Gro,GT}, but also in Conformal Field Theory \cite{Dou} and in Collective Field Theory \cite{GM}. Since then, mathematicians tried to derive rigorously some of the formulae used by these physicists, for instance \cite{AS,BS,DN,DGHK,Hal,Lev3,Lev,LM,Sen4,Sen3,Sen,Sen2}. We will focus here on the asymptotics of partition functions of the two-dimensional Yang--Mills model over a compact surface, written as sums over irreducible characters of the structure group. Depending on the orientability and genus of the underlying surface, we will link the limit of the partition function to several special functions from number theory and combinatorics: the Witten zeta function, the Jacobi theta function and the Euler function. 

\subsection{The Yang--Mills partition function on a compact surface}

Let $\lambda=(\lambda_{1}\geq \ldots \geq \lambda_{N})\in \Z^{N}$ be a non-increasing sequence of relative integers. We associate two real numbers to $\lambda$ : the {\em dimension}
\begin{equation}\label{eq:dim}
d_{\lambda}=\prod_{1\leq i<j\leq N} \frac{\lambda_i-\lambda_j+j-i}{j-i} = \prod_{1\leq i<j\leq N} \left(1+\frac{\lambda_i-\lambda_j}{j-i}\right),
\end{equation}
which is indeed a positive integer, and the {\em quadratic Casimir number}
\begin{equation}\label{eq:CasU}
c_2(\lambda) = \frac{1}{N}\left(\sum_{i=1}^N \lambda_i^2 + \sum_{1\leq i<j\leq N} (\lambda_i-\lambda_j)\right).
\end{equation}
These definitions are dictated by the representation theory of the unitary group $\U(N)$: the $N$-tuple $\lambda$, which we will also call a {\em highest weight} in this paper, labels (up to isomorphism) an irreducible representation of $\U(N)$ with dimension $d_{\lambda}$, and on which the Casimir operator of $\U(N)$, that is, the Laplacian, acts by the scalar $-c_{2}(\lambda)$. We will use the notation
\[\widehat \U(N)=\{(\lambda_{1}, \ldots , \lambda_{N}) \in \Z^{N} : \lambda_{1}\geq \ldots \geq \lambda_{N}\}.\]

Throughout this article, a \emph{surface} will denote a \emph{compact connected closed surface}. A standard classification theorem (see in \cite{Mas} for instance) states that it is homeomorphic to either of the following:
\begin{enumerate}
\item The connected sum of $g$ tori\footnote{If $g=0$ then by convention it is a sphere; otherwise it can also be seen as a torus with $g$ handles.},
\item The connected sum of $g$ projective planes.
\end{enumerate}
In the first case, the surface is said to be \emph{orientable}, otherwise it will be \emph{non-orientable}, and in either case we will call $g$ the \emph{genus} of the surface. Given an orientable surface $\Sigma_{g,T}$ of genus $g\geq 0$ and total area $T\geq 0$, the \emph{partition function} of the Yang--Mills theory on $\Sigma_{g,T}$ with structure group $\U(N)$ is defined\footnote{Here, we take it as a definition; however, it follows from lattice gauge theory axioms and is derived for example in \cite{FMS,GT,Wit}.} by
\begin{align} \label{eq01}
Z_{N}(g,T)= \sum_{\lambda \in \widehat\U(N)}e^{-\frac{T}{2}c_2(\lambda)} d_{\lambda}^{2-2g}.
\end{align}
If $\Sigma_{g,T}^-$ is a non-orientable compact surface of area $T$ homeomorphic to the connected sum of $g$ projective planes, then the partition function on $\Sigma_{g,T}^-$ with structure group $\U(N)$ is defined\footnote{See \cite{Wit} for an explanation of the formula.} by
\begin{equation} \label{eq:pfno}
Z_N^-(g,T) = \sum_{\lambda\in\widehat{\U}(N)} e^{-\frac{T}{2}c_2(\lambda)}d_\lambda^{2-g} (\iota_\lambda)^g,
\end{equation}
where $\iota_\lambda$ is the \emph{Frobenius--Schur indicator} of an irreducible representation of $\U(N)$ with highest weight $\lambda$, given by
\[
\iota_\lambda = \int_{\U(N)}\chi_\lambda(g^2)\mathrm{d}g.
\]

These partition functions admit a special unitary variant, which differs from them in two aspects: the summation is restricted to the $N$-tuples $\lambda=(\lambda_{1}\geq \ldots \geq \lambda_{N-1}\geq \lambda_{N}=0)$ of which the last element is $0$, and the Casimir number is replaced by its special unitary version
\begin{equation}\label{eq:CasSU}
c'_{2}(\lambda)= \frac{1}{N}\left(\sum_{i=1}^N \lambda_i^2 -\frac{1}{N} \left(\sum_{i=1}^N \lambda_i\right)^2 + \sum_{1\leq i<j\leq N} (\lambda_i-\lambda_j)\right).
\end{equation}
It is worth emphasizing that $c'_{2}(\lambda)$ is a non-negative real number. Indeed, an application of the Cauchy--Schwarz inequality shows that the first sum is larger than the absolute value of the second, and the third one is non-negative by definition of $\lambda$. We introduce
\[\widehat \SU(N)=\{(\lambda_{1}, \ldots , \lambda_{N}) \in \Z^{N} : \lambda_{1}\geq \ldots \geq \lambda_{N}=0\},\]
which is in bijection with the irreducible representations of $\SU(N)$, and define
\begin{align} \label{eq02}
Z'_{N}(g,T)= \sum_{\lambda \in \widehat\SU(N)}e^{-\frac{T}{2}c'_2(\lambda)} d_\lambda^{2-2g},
\end{align}
\begin{align} \label{eq:pfnosu}
Z_{N}^{'-}(g,T)=  \sum_{\lambda\in\widehat{\SU}(N)} e^{-\frac{T}{2}c'_2(\lambda)}d_\lambda^{2-g} (\iota'_\lambda)^g,
\end{align}
where $\iota'$ is the Frobenius--Schur indicator of an irreducible representation of $\SU(N)$ with highest weight $\lambda$. Let us notice that when $T=0$, the summands in the cases of $\U(N)$ and $\SU(N)$ are the same, and only the set of summation changes. In this `zero-temperature' situation, the partition function $Z'_{N}(g,0)$ was already studied by Witten \cite{Wit}, and later by Zagier \cite{Zag} who called it \emph{Witten zeta function} and denoted it by $\zeta_{\su(N)}(2g-2)$. The denomination `zeta function' comes from the remark that $Z'_{2}(g,0)= \zeta(2g-2)$, where $\zeta$ is the Riemann zeta function.

\subsection{Statement of the results}

The purpose of this paper is to establish the large $N$ limit of the partition functions, depending on the orientability and the genus of the underlying surface. The two following theorems state the limit of orientable surfaces of genus 1 and higher, and non-orientable surfaces of genus 2 and higher. The next sections will then be devoted to prove these theorems. Before stating the theorems, let us recall two special functions.
\begin{itemize}
\item The \emph{Jacobi theta function} $\vartheta$ is defined, for $(z,\tau)\in\C^2$ such that $\Im(\tau)>0$:
\[\vartheta(z;\tau) =  \sum_{n\in \Z} e^{i\pi n^2\tau + 2i\pi nz}.\]
We will also denote a particular version of this function by $\theta(q)=\sum_{n\in\Z} q^{n^2}$ for $q\in\C$ such that $|q|<1$.
\item The \emph{Euler function} $\phi$ is defined, for $q\in\C$ such that $|q|<1$, as the infinite product
\[
\phi(q)=(q;q)_\infty=\prod_{m=1}^\infty(1-q^m).
\]
\end{itemize}
\begin{theorem}[Orientable limits]\label{thm:main} Let $\Sigma$ be an orientable surface of genus $g$ and area $T\geq 0$. Set $q=e^{-\tfrac{T}{2}}$.
\begin{enumerate}
\item If $g\geq 2$ and $T>0$, then the following convergences hold:
\[\lim_{N\to \infty} Z_{N}(g,T)=\theta(q) \ \text{ and } \ \lim_{N\to \infty} Z'_{N}(g,T)=1.\]
Moreover, if $g\geq 2$ and $T=0$, we have
\[\lim_{N\to \infty} Z'_{N}(g,0)=1.\]
\item If $g=1$ and $T>0$, then the following convergences hold:
\[\lim_{N\to \infty} Z_{N}(1,T)=\frac{\theta(q)}{\phi(q)^2} \ \text{ and } \ \lim_{N\to \infty} Z'_{N}(1,T)=\frac{1}{\phi(q)^2}.\]
\end{enumerate} 
\end{theorem}

\begin{theorem}[Non-orientable limits]\label{thm:mainno}
Let $\Sigma$ be a non-orientable surface of genus $g\geq 2$ and area $T\geq 0$. Set $q=e^{-\tfrac{T}{2}}$.
\begin{enumerate}
\item If $g\geq 3$ and $T\geq 0$, then the following convergences hold:
\[\lim_{N\to \infty} Z_{N}^-(g,T)=Z_{N}^{'-}(g,T)=1.\]
\item If $g=2$ and $T>0$, then the following convergences hold:
\[\lim_{N\to \infty} Z_{N}^-(2,T)=\lim_{N\to \infty} Z_{N}^{'-}(2,T)=\frac{1}{\phi(q^2)}.\]
\end{enumerate}
\end{theorem}

\subsection{Comparison with other results}

In our setting, we only consider the case of compact connected closed surfaces with a fixed area $T$. In the orientable cases, this was already studied by Gross alone \cite{Gro} as well as with Taylor \cite{GT}: they found some results that we generalize here (see Section \ref{sec:torus} for more details). The limit given in Theorem \ref{thm:main} for $g=1$ is also mentioned without proof in \cite[eq.(3.2)]{Dou}, and the limit for $g>1$ is in adequation with a result by Rusakov \cite{Rus}. Let us remark that Rusakov affirmed that there is a nonzero free energy in the case of the torus, which is contradicted by the point $(ii)$ of Theorem \ref{thm:main}. The asymptotic behaviour of the partition function on the sphere is very different from the higher genus surfaces, and needs more analytical tools. Its free energy was computed rigorously by Boutet de Monvel and Shcherbina \cite{BS} and later by L\'evy and Ma\"ida \cite{LM}, as well as Dahlqvist and Norris \cite{DN}; in particular, L\'evy and Ma\"ida proved a phase transition conjectured by Douglas and Kazakov \cite{DK}. We do not consider this case because it needs different tools than the ones we use. We also leave aside the case of a non-orientable surface of genus $1$, which is homeomorphic to the projective plane, because our tools do not provide any concluding result; we expect it to be more closely related to the case of the sphere, because the dimensions of the irreducible representations are raised to a positive power.

Although we do not consider surfaces with boundaries, there are plenty of works, in particular in physics, devoted to the partition function on a cylinder. For instance, Gross and Matytsin \cite{GM} conjectured that there might be the same kind of phase transition for the cylinder with fixed boundary conditions as the one happening for the sphere. However, they used non-rigorous techniques leading to the asymptotic estimation of irreducible characters of $\SU(N)$ for which Tate and Zelditch \cite{TZ} exhibited a counterexample. Zelditch then obtained in \cite{Zel} a different result, using the so-called MacDonald identities, and computed the free energy on a cylinder with area $\tfrac{T}{N}$. In the mathematical literature, Guionnet and Ma\"ida \cite{GM2} developed some character expansion techniques that applied in this setting. Note here that the scaling regime is different from ours, and it might be enlightening to see how the limits of partition functions change when the area depends on $N$ as in the works of Zelditch or Guionnet--Ma\"ida.

\subsection{General remarks}

Before diving into the proofs of Theorems \ref{thm:main} and \ref{thm:mainno}, let us state a few facts that we find interesting around these Theorems.

\begin{itemize}
\item The limit of the partition function in the unitary case for $g\geq 2$ is the common value of $Z_{1}(g,T)$ for all $g\geq 0$. Indeed, the irreducible representations of $\U(1)$ are indexed by integers $n\in\mathbb{Z}$, and as $\U(1)$ is abelian, they all have dimension 1. Moreover, the Casimir number $c_2(n)$ is simply equal to $n^2$, therefore the partition function $Z_1(g,T)$ can be written
\begin{align*}
Z_{1}(g,T) = \sum_{n\in\mathbb{Z}} e^{-\frac{T}{2}n^2}=\theta(q)
\end{align*}
as expected. It could also be said that the limiting value of the partition functions $Z'_{N}(g,T)$ is also the value $Z'_{1}(g,T)$, understood as the partition function associated with the trivial group $\SU(1)$, with a unique irreducible representation of dimension $1$ and Casimir number $0$.
\item In the case of orientable surfaces, it appears that the limits of $Z_N$ and $Z'_N$ always differ from one factor, which is actually $Z_1(g,T)$. We can summarize this asymptotic factorization as follows:
\begin{equation}
\lim_{N\to\infty} Z_N(g,T)=\lim_{N\to\infty}Z_1(g,T)Z'_N(g,T).
\end{equation}

\item Numerical simulations suggest that for all $g\geq 2$ and all $T\geq 0$, the sequences $(Z_{N}(g,T))_{N\geq 2}$ and $(Z'_{N}(g,T))_{N\geq 2}$ might be non-increasing. This would be an interesting fact, that we are not yet able to establish.

\item Using the Jacobi triple product formula
\[
\sum_{n\in\Z}q^{n^2} = \prod_{m=1}^\infty (1-q^{2m})(1+q^{2m-1})^2,
\]
the limit of $Z_N(1,T)$ can be rewritten as an infinite product:
\[\lim_{N\to \infty} Z_{N}(1,T)=\prod_{m=1}^{\infty} \frac{(1+q^{m})(1+q^{2m-1})^{2}}{1-q^{m}}.\]
It does not particularly enlightens the nature of the limit but it makes it at least easier to approximate numerically.
\end{itemize}

We now turn to the proofs of Theorem \ref{thm:main}, which is given in Section \ref{sec:orient}, and Theorem \ref{thm:mainno}, which is given in Section \ref{sec:norient}.

\section{Orientable surfaces}\label{sec:orient}
\subsection{Orientable surfaces of genus $g\geq 2$}
\subsubsection*{The special unitary case}

We will start by proving Theorem \ref{thm:main}.(i) in the special unitary case. Let us first reduce the problem to the case where $T=0$ and $g=2$.

\begin{lemma}\label{lem:T0} For all $g\geq 0$, all $T\geq 0$, and all $N\geq 1$, we have
\[1\leq Z'_{N}(g,T)\leq Z'_{N}(2,0).\]
\end{lemma}

It follows from this lemma that the special unitary case of Theorem \ref{thm:main}.(i) is implied by the assertion
\begin{equation}\label{SUT=0}
\lim_{N\to \infty} Z'_{N}(2,0)=1,
\end{equation}
which we will prove in this section.

\begin{proof}[Proof of Lemma \ref{lem:T0}] The $N$-tuple $(0,\ldots,0)$ has dimension $1$ and Casimir number $0$. Thus, it contributes $1$ to the partition function $Z'_{N}(g,T)$, which explains the first inequality. The second inequality is an immediate consequence of the fact that all Casimir numbers are non-negative, and that all dimensions are positive integers.
\end{proof}

Our goal is now to prove \eqref{SUT=0}. We will deduce it from the following fact about Witten zeta functions. 

\begin{proposition}\label{prop:upperbound} For all real $s>1$, one has
\[\sup_{N\geq 1} \zeta_{\su(N)}(s)=\sup_{N\geq 1} \sum_{\lambda\in \widehat\SU(N)} d_{\lambda}^{-s}<\infty.\]
More precisely,
\[\lim_{N\to \infty} \zeta_{\su(N)}(s)=1 \ \text{ and } \ \lim_{N\to \infty}\sum_{\substack{\lambda\in \widehat\SU(N)\\ \lambda\neq(0,\ldots,0)}} d_{\lambda}^{-s}=0.\]
\end{proposition}

The proof of this proposition relies on three lemmas.

\begin{lemma} \label{lem:borneprod} For all $s>1$ and all $N\geq 1$, one has
\begin{equation}\label{eq:bp}
\sum_{\lambda\in \widehat\SU(N)} d_{\lambda}^{-s} \leq \prod_{k=1}^{N-1} \sum_{n\geq k} \binom{n}{k}^{-s}.
\end{equation}
\end{lemma}

\begin{proof} Let us choose $s>1$ and $N\geq 1$. In the left-hand side of \eqref{eq:bp}, which is a sum over $\lambda_{1}\geq \ldots \lambda_{N}\geq 0$, let us make the change of variables
\[m_{1}=\lambda_{1}-\lambda_{2}+1, \ldots, m_{N-1}=\lambda_{N-1}-\lambda_{N}+1.\]
The new variables $m_{1},\ldots,m_{N-1}$ are now independent, and positive. Using \eqref{eq:dim}, we find
\begin{equation}\label{eq:dimm}
d_{\lambda}=\prod_{1\leq i<j \leq N} \frac{m_{i}+\ldots+m_{j-1}}{j-i},
\end{equation}
so that
\begin{align*}
\sum_{\lambda_{1}\geq \ldots \geq \lambda_{N}= 0} d_{\lambda}^{-s}&=\sum_{m_{1},\ldots,m_{N-1}\geq 1} \ \prod_{1\leq i < j \leq N} \frac{(j-i)^{s}}{(m_{i}+\ldots+m_{j-1})^{s}}\\
&=\sum_{m_{1},\ldots,m_{N-1}\geq 1}\  \prod_{k=1}^{N-1} \prod_{i=1}^{k} \frac{(k-i+1)^{s}}{(m_{i}+\ldots+m_{k})^{s}} & (k=j-1)
\end{align*}
Since $m_{i}+\ldots+m_{k-1}\geq k-i$, we obtain
\begin{align*}
\sum_{\lambda_{1}\geq \ldots \geq \lambda_{N}= 0} d_{\lambda}^{-s}&\leq \sum_{m_{1},\ldots,m_{N-1}\geq 1} \prod_{k=1}^{N-1} \prod_{i=1}^{k} \frac{(k-i+1)^{s}}{(m_{k}+k-i)^{s}}\\
&= \sum_{m_{1},\ldots,m_{N-1}\geq 1} \prod_{k=1}^{N-1} \binom{k+m_{k}-1}{k}^{-s}\\
&=\prod_{k=1}^{N-1} \sum_{n\geq k} \binom{n}{k}^{-s},
\end{align*}
which is the announced upper bound.
\end{proof}

\begin{lemma}\label{lem:binom} For all real $s>1$,
\[\sum_{k\geq 1}\sum_{n>k}  \binom{n}{k}^{-s}<\infty.\]
\end{lemma}
%\sum_{n=2}^{\infty}\sum_{k=1}^{n-1}

\begin{proof} We use the fact that for $k$ between $2$ and $n-2$, the inequality $\binom{n}{k}\geq \binom{n}{2}$ holds. Thus, 
\[\sum_{k\geq 1}\sum_{n>k}  \binom{n}{k}^{-s}\leq 2^{-s}+\sum_{n=3}^{\infty}\bigg(\frac{2}{n^{s}}+(n-3) \frac{2^{s}}{n^{s}(n-1)^{s}}\bigg),\]
which is indeed finite for $s>1$.
\end{proof}

\begin{lemma} \label{lem:min_dim}
Let $\lambda$ be an element of $\widehat\SU(N)$. If $\lambda=(0,\ldots,0)$, then $d_{\lambda}=1$. Otherwise, $d_{\lambda}\geq N$.
\end{lemma}

\begin{proof} Let us use again the variables $m_{1},\ldots,m_{N-1}$ introduced in the proof of Lemma \ref{lem:borneprod}. It is manifest on the expression \eqref{eq:dimm} of $d_{\lambda}$ that this dimension is increasing in each of the variables $m_{1},\ldots,m_{r}$. The case where each of these variables is equal to $1$ is the case where $\lambda=(0,\ldots,0)$ and $d_{\lambda}=1$. Any other irreducible representation has a dimension that is at least equal to the dimension of one of the representations 
\[\lambda_{1}=(1,0,\ldots,0), \lambda_{2}=(1,1,0,\ldots,0), \ldots, \lambda_{N-1}=(1,\ldots,1,0).\]
These representations, which are the exterior powers of the standard representation of $\SU(N)$, have dimensions 
\[d_{\lambda_{k}}=\binom{N}{k}\geq N, \ k\in \{1,\ldots,N-1\}.\]
Thus, $d_{\lambda}\geq N$, as expected.
\end{proof}

We can now prove Proposition \ref{prop:upperbound}.

\begin{proof}[Proof of Proposition \ref{prop:upperbound}] The bound obtained in Lemma \ref{lem:borneprod} can be rewritten as
\[\sum_{\lambda_{1}\geq \ldots \geq \lambda_{N}= 0} d_{\lambda}^{-s} \leq \prod_{k=1}^{N-1} \left[1+\sum_{n>k} \binom{n}{k}^{-s}\right]\leq \exp \sum_{k=1}^{\infty} \sum_{n>k} \binom{n}{k}^{-s}\]
and this last bound, independent of $N$, is finite by Lemma \ref{lem:binom}. This proves the first assertion. 

For the second, let us introduce a real $s'\in (1,s)$ and use Lemma \ref{lem:min_dim}. We find
\[\sum_{\substack{\lambda\in \widehat\SU(N)\\ \lambda\neq(0,\ldots,0)}}d_\lambda^{-s}\leq  N^{-(s-s')}\sum_{\lambda\in \widehat\SU(N)} d_\lambda^{-s'},\]
which tends to $0$ as $N$ tends to infinity. 
\end{proof}

In order to prove \eqref{SUT=0}, we need a last piece of information about the dimensions of the irreducible representations of $\SU(N)$.

\begin{proof}[Proof of Theorem \ref{thm:main}.(i) in the special unitary case]
On one hand, Lemma \ref{lem:T0} implies that $Z'_{N}(2,0)\geq 1$. On the other hand, 
\[Z'_{N}(2,0) = \sum_{\lambda\in \widehat\SU(N)} d_\lambda^{-2} = 1 + \sum_{\lambda\neq (0,\ldots,0)} d_\lambda^{-2}.\]
Using Lemma \ref{lem:min_dim}, we find
\[Z'_{N}(2,0) \leq  1 + N^{-\frac{1}{2}}\sum_{\lambda\neq (0,\ldots,0)} d_\lambda^{-\frac{3}{2}}\leq  1 + N^{-\frac{1}{2}}\sup_{N\geq 1}\sum_{\lambda\in \widehat\SU(N)} d_\lambda^{-\frac{3}{2}}.\]
Thanks to Proposition \ref{prop:upperbound}, this implies
\[\limsup_{N\to \infty}Z'_{N}(2,0)\leq 1\]
and this concludes the proof of \eqref{SUT=0}, hence of Theorem \ref{thm:main}.(i) in the special unitary case.
\end{proof}

\subsubsection*{The unitary case}

We treat the unitary case of Theorem \ref{thm:main}.(i) using our understanding of the special unitary case, and the bijection
%\begin{align*}
%\widehat\SU(N)\times \Z & \build{\longrightarrow}{}{\sim} \widehat\U(N)\\\
%(\lambda,n) & \longmapsto \lambda+n=(\lambda_{1}+n,\ldots,\lambda_{N}+n).
%\end{align*}
\begin{align*}
\Phi:\left\lbrace
\begin{array}{rcl}
\widehat{\SU}(N)\times \Z & \build{\longrightarrow}{}{\sim} & \widehat{\U}(N)\\
(\lambda,n) & \longmapsto & \lambda+n=(\lambda_{1}+n,\ldots,\lambda_{N}+n).
\end{array}\right.
\end{align*}
We will keep throughout this section the notation $\lambda$ for an element of $\widehat\SU(N)$, $n$ for an element of $\Z$, $\lambda+n$ for the corresponding element of $\widehat\U(N)$, and $|\lambda|=\lambda_1+\cdots+\lambda_N$. The first observation is the following.

\begin{lemma} \label{lem:cascas}
We have the equality
\[c_{2}(\lambda+n)=c_{2}'(\lambda)+\bigg(n+\frac{|\lambda|}{N}\bigg)^{2}.\]
\end{lemma}

\begin{proof} The proof is a simple verification using the definitions \eqref{eq:CasU} and \eqref{eq:CasSU} of $c_{2}$ and $c_{2}'$.
\end{proof}

It is the contribution of the highest weights of the form $0+n=(n,\ldots,n)$ which produces the Jacobi theta function in the unitary part of Theorem \ref{thm:main}.(i). We will prove that the contribution of all other elements of $\widehat\U(N)$ vanishes in the large $N$ limit.

\begin{proof}[Proof of Theorem \ref{thm:main}.(i) in the unitary case] Let us consider $g\geq 2$ and $T>0$, and set $q=e^{-\tfrac{T}{2}}$. We split the partition function $Z_{N}(g,T)$ into two parts
\[Z_{N}(g,T)=\sum_{n\in \Z} q^{n^{2}} +\sum_{\substack{\lambda\in \widehat\SU(N)\\ \lambda\neq (0,\ldots,0)}}\sum_{n\in \Z} q^{c_{2}(\lambda+n)}d_{\lambda+n}.\]
The first part corresponds to highest weights of the form $(n,\ldots,n)$, which have Casimir numbers $n^{2}$ and dimension $1$, and is equal to $\theta(q)$. The second part is the contribution of all the other highest weights. To compute it, we observe that $d_{\lambda+n}=d_{\lambda}$ and we use Lemma \ref{lem:cascas}.
We find
\[0\leq Z_{N}(g,T)-\theta(q)\leq\sum_{\substack{\lambda\in \widehat\SU(N)\\ \lambda\neq (0,\ldots,0)}} \bigg(\sum_{n\in \Z} q^{(n+|\lambda|/N)^{2}}\bigg)q^{c_{2}'(\lambda)}d_{\lambda}^{2-2g}.
\]
The sum between the brackets is bounded independently of $N$, for example, in a very elementary way, by $C=1+\theta(q)$. Hence, the right-hand side is bounded by
\[C\sum_{\substack{\lambda\in \widehat\SU(N)\\ \lambda\neq (0,\ldots,0)}} d_{\lambda}^{2-2g}=C\big(\zeta_{\su(N)}(2g-2)-1\big)\]
which, thanks to Proposition \ref{prop:upperbound}, converges to $0$.
\end{proof}

\subsection{The torus}\label{sec:torus}

Our proof of the convergence of the partition function when $g\geq 2$ was based on our study of the dimensions of the irreducible representations of $\SU(N)$, expressed in Proposition \ref{prop:upperbound}. A glance at \eqref{eq01} shows that when $g=1$, these dimensions do not appear anymore in the partition function, and to treat this case we need to use completely different estimates. In this section, we will prove that $Z_N(1,T)$ still admits a finite limit for $T>0$, but this limit turns out to be different: it will involve the classical generating function of integer partitions. Recall that if we denote, for each $n\geq 0$, by $p(n)$ the number of partitions of the integer $n$, we have the equality of formal series in the variable $q$:
\begin{equation}\label{eq:GFpart}
\sum_{n\geq 0} p(n)q^{n} = \prod_{m=1}^{\infty}\frac{1}{1-q^{m}}=\phi(q)^{-1}.
\end{equation}

Before entering the technical details, let us explain the idea of the proof of Theorem \ref{thm:main}.(ii), at least in the special unitary case. In the present situation where $g=1$, the partition function is
\[Z'_{N}(1,T)=\sum_{\lambda\in\widehat\SU(N)}e^{-c'_{2}(\lambda)\frac{T}{2}}=\sum_{\lambda\in\widehat\SU(N)}q^{c'_{2}(\lambda)},\]
using the notation $q=e^{-\frac{T}{2}}$. The problem is thus to identify which highest weights of $\SU(N)$ keep, in the large $N$ limit, a bounded quadratic Casimir number, and bring a non-zero contribution to the partition function. We claim, although this statement is not very precise at this stage, that these highest weights are those depicted in Fig. \ref{fig02} (in the special unitary case, we need to look at the right part of this figure). They are the highest weights that are flat up to a small\footnote{Small compared to $N$ but not necessarily finite.} perturbation at each end, represented by two partitions $\alpha$ and $\beta$ of length $\leq N/2$. Let us call these highest weights {\em almost flat}. A similar description was proposed by Gross and Taylor in \cite{GT}, but in the case where the perturbations remain finite, and their goal was rather to obtain a $1/N$ expansion of the partition function than to find its large $N$ limit. The smaller the length of $\alpha$ and $\beta$, the flatter the highest weight: typically we will consider $\alpha$ and $\beta$ of length $\ll \sqrt{N}$. Using the notation $\lambda(\alpha,\beta)$ introduced in Fig. \ref{fig02}, and the notation $|\alpha|$ (resp. $|\beta|$) for sum of the components of $\alpha$ (resp. $\beta$), the main estimate will be a refinement of the equality
\[c_{2}'(\lambda(\alpha,\beta)) = |\alpha|+|\beta| +O(N^{-1}) \]
with an explicit expression of the error in terms of $\alpha$ and $\beta$. The outline of the proof is then the following:
\[Z'_{N}(1,T)\simeq \sum_{\substack{\lambda \in \widehat\SU(N)\\ \lambda \text{ almost flat}}} q^{c'_{2}(\lambda)}\simeq\sum_{\alpha,\beta \text{ of length }\ll \sqrt{N}} q^{c'_{2}(\lambda(\alpha,\beta))}\simeq\sum_{\alpha,\beta \text{ of length }\ll \sqrt{N}} q^{|\alpha|+|\beta|}
\]
and the last sum tends to the square of the generating function of integer partitions when $N\to\infty$.

\subsubsection*{Almost flat highest weights}

From two integer partitions $\alpha=(\alpha_{1}\geq \cdots \geq \alpha_{r}> 0)$ and $\beta=(\beta_{1}\geq \cdots \geq \beta_{s}> 0)$ of respective lengths $r$ and $s$, and an integer $n\in \Z$, we can form, for all $N\geq r+s+1$, the highest weight
\[\lambda_{N}(\alpha,\beta,n)=(\alpha_1+n,\ldots,\alpha_r+n,\underbrace{n,\ldots,n}_{N-r-s},n-\beta_s,\ldots,n-\beta_1) \in \widehat\U(N),\]
which we also denote by $\lambda(\alpha,\beta,n)$ when there is no doubt on the value of $N$. We extend this definition in the obvious way to the cases where one or both of the partitions $\alpha$ and $\beta$ are the empty partition.

We can also form the highest weight
\[\lambda_{N}(\alpha,\beta)=\lambda_{N}(\alpha,\beta,\beta_{1})\in \widehat\SU(N),\]
with the convention that $\lambda_{N}(\alpha,\varnothing)=\lambda_{N}(\alpha,\varnothing,0)=(\alpha_{1},\ldots,\alpha_{r},0)$.

These constructions are illustrated in Fig. \ref{fig02} below. The reader may have noticed that the definition of $\lambda_{N}(\alpha,\beta,n)$ still makes sense when $N=r+s$ and wonder why we exclude this case. The reason is that under the stronger assumption $N\geq r+s+1$, it is possible to recover $\alpha$ and $\beta$ unambiguously from the data of $\lambda_{N}(\alpha,\beta,n)$, $r$ and $s$. A counterexample with $r=s=1$ and $N=2$ is given in Fig. \ref{fig:nonu}. Without the data of $r$ and $s$, there are usually multiple ways of writing a highest weight in the form $\lambda_{N}(\alpha,\beta,n)$, see also Fig.~\ref{fig:nonu}. Finally, it should be emphasized that every highest weight of $\U(N)$ or $\SU(N)$ can be written as $\lambda_{N}(\alpha,\beta,n)$ or $\lambda_{N}(\alpha,\beta)$.

\begin{figure}[!h]
\centering
\includegraphics[scale=1]{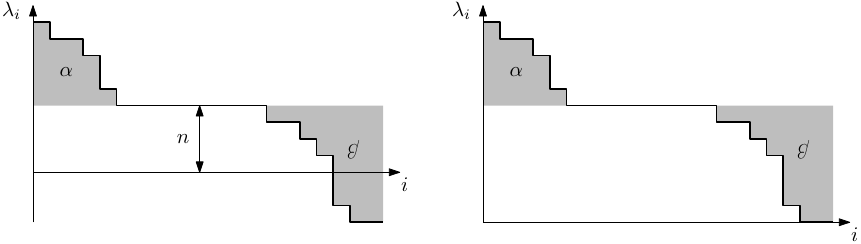}
\caption{\small From two partitions $\alpha$ and $\beta$ and an integer $n\in \Z$, we can form the highest weights $\lambda(\alpha,\beta,n)\in \widehat\U(N)$ (on the left) and $\lambda(\alpha,\beta)\in \widehat{\SU}(N)$ (on the right).}\label{fig02}
\end{figure}

The Casimir number of a highest weight can be expressed conveniently through this decomposition, as we will show below. First, let us recall the definition of the content of a box of a diagram, which is mentioned in particular in \cite{Sta,OV}.

\begin{figure}[h!]
\centering
\includegraphics[scale=0.8]{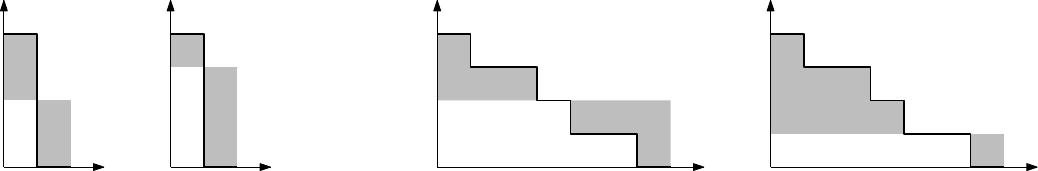}
\caption{\small On the left: the highest weight $(4,0)$ can be written in several ways as $\lambda_{2}(\alpha,\beta)$ with $\alpha$ and $\beta$ of length $1$. On the right: the highest weight $(4,3,3,2,1,1,0)$ is equal to $\lambda_{7}((2,1,1),(2,1,1))$ as well as to $\lambda_{7}((3,2,2,1),(1))$.}\label{fig:nonu}
\end{figure}

\begin{definition}
Let $\alpha=(\alpha_1\geq\cdots\geq\alpha_r\geq 0)$ be a non-increasing sequence of integers, seen as a Young diagram. For any box $(i,j)$ of this diagram, that is, any $(i,j)$ such that $j\leq \alpha_{i}$, we call \emph{content} of the box $(i,j)$ the quantity $c(i,j)=j-i$. We also define the \emph{total content} $K(\alpha)$ of $\alpha$ as the sum of the contents of the boxes of $\alpha$.
\end{definition}

An example is given on Fig. \ref{fig01}.

\begin{figure}[!h]
\begin{center}
\includegraphics{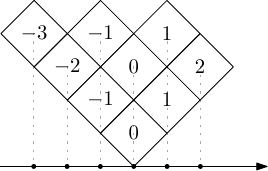}
\end{center}
\caption{Filling of the boxes of $(3,3,2,1)$ with their respective contents. The Young diagram is represented here in the so-called Russian way, where the content of a box is its abscissa.}\label{fig01}
\end{figure}

The main result of this section is the following.

\begin{proposition}\label{prop04}
Let $\alpha$ and $\beta$ be two partitions of respective lengths $r$ and $s$. Let $n$ be an integer. Then, provided $N\geq r+s$, we have
\begin{equation}\label{eq15}
c_2(\lambda(\alpha,\beta,n)) = \vert\alpha\vert + \vert\beta\vert + n^2 +  \frac{2}{N}\big(K(\alpha)+K(\beta)+n(\vert\alpha\vert-\vert\beta\vert)\big)
\end{equation}
in the unitary case, and
\begin{equation}\label{eq16}
c'_2(\lambda(\alpha,\beta)) = \vert\alpha\vert + \vert\beta\vert + \frac{2}{N}(K(\alpha)+K(\beta)) + \frac{1}{N^2}\left(\vert\alpha\vert-\vert\beta\vert\right)^2
\end{equation}
in the special unitary case.
\end{proposition}

\begin{proof}
Let us start with the unitary case. Using the definition of Casimir number and the definition of $\lambda(\alpha,\beta,n)$, we obtain
\begin{align*}
Nc_2(\lambda(\alpha,\beta,n)) = & \sum_{i=1}^r \alpha_i^2 + \sum_{1\leq i<j\leq r} (\alpha_i-\alpha_j) + 2 n\vert \alpha \vert+\sum_{i=1}^s \beta_i^2 + \sum_{1\leq i <j\leq s} (\beta_i-\beta_j) - 2n\vert\beta\vert\\
&+\vert\alpha\vert (N-r-s) + \vert\beta\vert (N-r-s) + \sum_{\substack{1\leq i\leq r\\ 1\leq j\leq s}} (\alpha_i+\beta_j)+Nc^2\\
= & N(\vert\alpha\vert +\vert\beta\vert + n^2)+2n(\vert\alpha\vert-\vert\beta\vert)+\sum_{i=1}^r \alpha_i^2 + \sum_{1\leq i<j\leq r} (\alpha_i-\alpha_j)-r\vert\alpha\vert\\
&\hspace{5.48cm}+\sum_{i=1}^s \beta_i^2 + \sum_{1\leq i<j\leq s} (\beta_i-\beta_j)-s\vert \beta\vert.
\end{align*}
On the other hand,
\[K(\alpha)=\sum_{i=1}^{r} \frac{\alpha_{i}(\alpha_{i}+1)}{2}-i\alpha_{i}=\frac{1}{2}\left(\sum_{i=1}^r \alpha_i^2 + \sum_{1\leq i<j\leq r} (\alpha_i-\alpha_j)-r\vert\alpha\vert\right)\]
and we find \eqref{eq15} as announced.

Concerning the special unitary case, we simply need to subtract from $c_2(\lambda)$ the quantity $\frac{1}{N^2}\left(\sum \lambda_i\right)^2$, which leads to
\begin{align*}
c'_2(\lambda(\alpha,\beta)) = & c_2(\lambda(\alpha,\beta,\beta_{1}))-\frac{1}{N^2}\left(\vert\alpha\vert-\vert\beta\vert+N\beta_1\right)^2
\end{align*}
from which \eqref{eq16} follows easily.
\end{proof}

\subsubsection*{The special unitary case}

In our treatment of the special unitary case, we want to adopt a systematic way of writing a highest weight of $\SU(N)$ under the form $\lambda_{N}(\alpha,\beta)$. We do this in a way that depends on the parity of $N$, but that in any case rests on the observation that for all $M_{1},M_{2}\geq 0$, the map
\begin{align*} 
\widehat\SU(M_{1}+1)\times\widehat\SU(M_{2}+1) & \stackrel{\sim}{\longrightarrow} \widehat\SU(M_{1}+M_{2}+1)\\
(\alpha,\beta) & \longmapsto \lambda_{M_{1}+M_{2}+1}(\alpha,\beta)
\end{align*}
is a bijection. 

In the case where $N$ is odd, equal to $2M+1$, we take $M_{1}=M_{2}=M$. When $N=2M$ is even, and positive, we choose $M_{1}=M-1$ and $M_{2}=M$. In this section, we will always write highest weights of $\SU(N)$ as $\lambda(\alpha,\beta)$, and this will always refer to the decomposition just described.

The proof of Theorem \ref{thm:main}.(ii) will rely on two estimates of the Casimir number: one that helps proving the convergence of the sum of $q^{c'_2(\lambda)}$ over almost flat highest weights $\lambda$ to the expected limit, and one that helps controlling the sum over remaining highest weights. 

\begin{lemma}\label{lem:su_odd}
Let $\lambda=\lambda(\alpha,\beta)\in\widehat{\SU}(N)$. Set $k=\vert\alpha\vert+\vert\beta\vert$. Then the following inequalities hold:
\begin{equation}\label{eq14}
k-\frac{k^2}{N} \leq c'_2(\lambda(\alpha,\beta))\leq k+\frac{k^2}{N}+\frac{k^2}{N^2},
\end{equation}
\begin{equation}\label{eq18}
\frac{k}{2} \leq c'_2(\lambda(\alpha,\beta)).
\end{equation}
\end{lemma}

\begin{proof} We start from the expression of $c'_{2}(\lambda(\alpha,\beta))$ given by \eqref{eq16}. The point is to bound $K(\alpha)$ and $K(\beta)$. 

The list of the contents of the boxes of $\alpha$ taken row after row and from left to right in each row (as on the left of Fig. \ref{fig:sens}) is a sequence $x_{1},\ldots,x_{|\alpha|}$ such that $|x_{i}|\leq i-1$ for each $i\in \{1,\ldots,|\alpha|\}$. It follows that
\[-|\alpha|(|\alpha|-1)\leq 2K(\alpha) \leq |\alpha|(|\alpha|-1).\]
This implies immediately
\[2 |K(\alpha)+K(\beta)| \leq k^{2},\]
and \eqref{eq14}, after observing that $0 \leq (\vert\alpha\vert -\vert\beta\vert)^2 \leq (\vert\alpha\vert + \vert\beta\vert)^2 = k^2$.

\begin{figure}[h!]
\begin{center}
\includegraphics{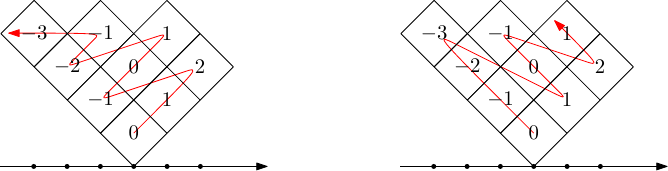}
\caption{\label{fig:sens} Two ways of listing the contents of the boxes of the diagram $(3,3,2,1)$.}
\end{center}
\end{figure}

Let us turn to the proof of \eqref{eq18}. We will establish a different lower bound on $K(\alpha)$ and $K(\beta)$. For this, let us list the contents of the boxes of $\alpha$, now taken column after column and from top to bottom in each column (as on the right of Fig. \ref{fig:sens}). It is now a sequence $x_{1},\ldots,x_{|\alpha|}$ of integers that along each column of $\alpha$ decreases by $1$ at each step, and at each change of column jumps to a positive integer. The crucial point is that the height of the columns of $\alpha$ is bounded by the integer that we called $M_{1}$ at the beginning of this section, and that is in any case not greater than $\frac{N}{2}$. The contribution of each column is thus bounded below by $-\frac{N}{4}$ times the number of boxes in this column. It follows that
\[K(\alpha)\geq -\frac{N}{4}|\alpha|,\]
and a similar argument holds for $\beta$. The result follows again from \eqref{eq16}.
\end{proof}

%Before proving this lemma, we notice that the assumption made on $n$ guarantees that $\lambda$ is almost flat in the sense that $\alpha$ and $\beta$ have a number of nonzero entries which is $o(N)$. In fact, let $\tilde{r}=\max\lbrace k:\alpha_k>0\rbrace$ and $\tilde{s}=\max\lbrace k:\beta_k>0\rbrace$, it is pretty clear that $\vert\alpha\vert\geq \tilde{r}$ and $\vert\beta\vert\geq\tilde{s}$; it appears that the condition $n\leq N^\gamma\leq \sqrt{N}$ implies that $\tilde{r}+\tilde{s}\leq \sqrt{N}$. In particular, $\tilde{r}+\tilde{s}<\min(r,s)$ for $N\gg1$. Moreover, this shows that the parity of $N$ doesn't change the estimations, despite the fact that the values of $s$ differ: the idea behind that is the fact that only the nonzero entries of $\alpha$ and $\beta$ play a role on the convergence.

\begin{proof}[Proof of theorem \ref{thm:main}.(ii) in the special unitary case]
Let us fix a real $\gamma\in (0,\frac{1}{2})$. Let us split the set of highest weights of $\SU(N)$ in four disjoint subsets:
\begin{align*}
\Lambda_{N,1}&=\{\lambda(\alpha,\beta) : |\alpha|\leq  N^{\gamma}, |\beta|\leq  N^{\gamma}\},\\
\Lambda_{N,2}&=\{\lambda(\alpha,\beta) : |\alpha|>  N^{\gamma}, |\beta|\leq  N^{\gamma}\},\\
\Lambda_{N,3}&=\{\lambda(\alpha,\beta) : |\alpha|\leq  N^{\gamma}, |\beta|>  N^{\gamma}\},\\
\Lambda_{N,4}&=\{\lambda(\alpha,\beta) : |\alpha|>  N^{\gamma}, |\beta|>  N^{\gamma}\}.
\end{align*}
For each $i\in \{1,2,3,4\}$, we define
\[S'_{N,i}=\sum_{\lambda\in \Lambda_{N,i}} q^{c_{2}'(\lambda)},\]
so that
\[Z'_{N}(1,T)=S'_{N,1}+S'_{N,2}+S'_{N,3}+S'_{N,4}.\]
The set $\Lambda_{N,1}$ is the set of highest weights that we think of as being almost flat, and we will now prove, in a first step, that they bring the only non-zero contribution in the limit where $N$ tends to infinity. 

Let $\lambda(\alpha,\beta)$ be an element of $\Lambda_{N,1}$. Then thanks to \eqref{eq14}, we have
\begin{equation}\label{eq:encadre}
|\alpha|+|\beta|-4N^{2\gamma-1} \leq c'_2(\lambda(\alpha,\beta)) \leq |\alpha|+|\beta| + 4N^{2\gamma-1}+4N^{2\gamma-2}.
\end{equation}
For $N$ large enough, any partition of an integer not greater than $N^{\gamma}$ has less than $\frac{N}{2}$ positive parts. Thus, if $\alpha$ and $\beta$ are any two such partitions, the highest weight $\lambda_{N}(\alpha,\beta)$ is well defined, and belongs to $\Lambda_{N,1}$. Thus, for $N$ large enough,
\[S'_{N,1}= \sum_{\vert\alpha\vert,\vert\beta\vert\leq N^\gamma} q^{c'_2(\lambda(\alpha,\beta))}.\]
From \eqref{eq:encadre}, we deduce that
\[q^{4N^{2\gamma-1}+4N^{2\gamma-2}} \sum_{\vert\alpha\vert,\vert\beta\vert\leq N^\gamma}q^{\vert\alpha\vert+\vert\beta\vert} \leq S'_{N,1} \leq 
q^{-4N^{2\gamma-1}}\sum_{\vert\alpha\vert,\vert\beta\vert\leq N^\gamma}q^{\vert\alpha\vert+\vert\beta\vert}.\]
Since $2\gamma-1$ is negative, the powers of $q$ in front of the sums on either side tend to $1$ as $N$ tends to infinity. On the other hand, the sum over $\alpha$ and $\beta$ tends, as $N$ tends to infinity, to the square of the generating function of partitions. Hence,
\[\lim_{N\to \infty} S'_{N,1}=\lim_{N\to \infty} \sum_{\vert\alpha\vert,\vert\beta\vert\leq N^\gamma}q^{\vert\alpha\vert+\vert\beta\vert} =\bigg(\sum_{\alpha} q^{|\alpha|}\bigg)^{2}= \prod_{m=1}^\infty (1-q^m)^{-2}.\]

In a second step, we prove that the three other contributions to $Z'_{N}(1,T)$ vanish as $N$ tends to infinity. For this, we use \eqref{eq18}. Let us treat the case of $S'_{N,2}$,  the case of $S'_{N,3}$ being perfectly similar, and the case of $S'_{N,4}$ even simpler. We have
\begin{align*}
0\leq S'_{N,2}\leq \sum_{|\alpha|\leq N^{\gamma}, |\beta|>N^{\gamma}} q^{\frac{1}{2}(|\alpha|+|\beta|)}\leq \sum_{\alpha}q^{\frac{1}{2}|\alpha|}\sum_{|\beta|>N^{\gamma}} q^{\frac{1}{2}|\beta|}=\sum_{\alpha}q^{\frac{1}{2}|\alpha|} \sum_{k>N^{\gamma}} p(k) q^{\frac{k}{2}}.
\end{align*}
The first sum is finite, and the second, as a remainder of a convergent series, tends to $0$ as $N$ tends to infinity. This concludes the proof.
\end{proof}

\subsubsection*{The unitary case}

The proof of Theorem \ref{thm:main}.(ii) in the unitary case will rely on the same tools as the special unitary case, that is, the use of almost flat highest weights, combined with the bijection $\Phi:(\lambda,n)\mapsto\lambda+n$ introduced in Section 2.1. In particular, Lemma \ref{lem:cascas} will be of great help in order to control the convergence of $Z_N(1,T)$ using the convergence of $Z'_N(1,T)$.

\begin{proof}[Proof of Theorem \ref{thm:main}.(ii) in the unitary case]

Let $\lambda(\alpha,\beta)$ be an element of $\widehat{\SU}(N)$. Using Lemma \ref{lem:cascas} and Proposition \ref{prop04}, it appears that, for all $n\in\Z$,
\[c_2(\lambda(\alpha,\beta)+n)=c'_2(\lambda(\alpha,\beta))+\left(n+\frac{|\lambda(\alpha,\beta)|}{N}\right)^2=c'_2(\lambda(\alpha,\beta))+\left(n+\frac{|\alpha|-|\beta|}{N}+\beta_1\right)^2,\]
so that we can write, up to a change of index $n\leftarrow n-\beta_1$,
\begin{equation}\label{eq:lambdaplusn}
Z_N(1,T)=\sum_{\lambda(\alpha,\beta)\in\widehat{\SU}(N)}\left(\sum_{n\in\Z} q^{\left(n+\frac{|\alpha|-|\beta|}{N}\right)^2}\right)q^{c'_2(\lambda(\alpha,\beta))}.
\end{equation}

The main difference with the case of $\SU(N)$ is the sum over $n$ between the brackets, and we will need to control it in order to get the convergence.

Let $\gamma\in(0,\frac{1}{2})$, and the subsets $(\Lambda_{N,i})_{1\leq i\leq 4}$ of $\widehat{\SU}(N)$ as in the special unitary case. We define, for $1\leq i\leq 4$,
\[S_{N,i}=\sum_{\lambda\in\Lambda_{N,i}} \left(\sum_{n\in\Z} q^{\left(n+\frac{|\alpha|-|\beta|}{N}\right)^2}\right)q^{c'_2(\lambda)},\]
and we obtain the following decomposition:
\[Z_N(1,T)=S_{N,1}+S_{N,2}+S_{N,3}+S_{N,4}.\]

Let $\lambda(\alpha,\beta)$ be an element of $\Lambda_{N,1}$. From the fact that $\big\vert|\alpha|-|\beta|\big\vert \leq |\alpha|+|\beta|\leq 2N^\gamma$ we get
\begin{equation}\label{eq:encadre2}
n^2-4nN^{\gamma-1}\leq \left(n+\frac{|\alpha|-|\beta|}{N}\right)^2 \leq n^2+4nN^{\gamma-1}+4N^{2\gamma-2}.
\end{equation}

For the same reason as in the special unitary case, for $N$ large enough we have
\[S_{N,1}=\sum_{|\alpha|,|\beta|\leq N^\gamma}\left(\sum_{n\in\Z} q^{\left(n+\frac{|\alpha|-|\beta|}{N}\right)^2}\right)q^{c'_2(\lambda(\alpha,\beta))};\]
Then, Equations \eqref{eq:encadre} and \eqref{eq:encadre2} yield
\begin{equation}\label{eq:encadre3}
q^{4N^{2\gamma-1}+4N^{2\gamma-2}+4N^{2\gamma-2}}\left(\sum_{n\in\Z} q^{n^2+4nN^{\gamma-1}}\right) \sum_{\vert\alpha\vert,\vert\beta\vert\leq N^\gamma}q^{\vert\alpha\vert+\vert\beta\vert} \leq S_{N,1}
\end{equation}

and
\begin{equation}\label{eq:encadre4}
S_{N,1} \leq 
q^{-4N^{2\gamma-1}}\left(\sum_{n\in\Z} q^{n^2-4nN^{\gamma-1}}\right) \sum_{\vert\alpha\vert,\vert\beta\vert\leq N^\gamma}q^{\vert\alpha\vert+\vert\beta\vert}.
\end{equation}

The sums $\sum_{n\in\Z} q^{n^2\pm 4nN^{\gamma-1}}$ in both cases tend to $\sum_{n\in\Z} q^{n^2}$ by dominated convergence, because $\gamma-1<0$. The remaining terms in both inequalities \eqref{eq:encadre3} and \eqref{eq:encadre4} behave in the same way as in the proof of Theorem \ref{thm:main}.(ii) in the special unitary case. This proves that
\[\lim_{N\to\infty} S_{N,1} = \sum_{n\in\Z} q^{n^2} \prod_{m=1}^\infty \frac{1}{(1-q^m)^2}.\]

Now let us treat the cases of $\Lambda_{N,2}$, $\Lambda_{N,3}$ and $\Lambda_{N,4}$. The arguments are the same for the three of them, so we only choose to detail the case of $\Lambda_{N,2}$. We have, using Equation \eqref{eq18},
\[0\leq S_{N,2}\leq \sum_{|\alpha|\leq N^{\gamma}, |\beta|>N^{\gamma}}\left(\sum_{n\in\Z} q^{\left(n+\frac{|\alpha|-|\beta|}{N}\right)^2}\right) q^{\frac{1}{2}(|\alpha|+|\beta|)},\]
and the sum between brackets can be bounded independantly from $N,$ $|\alpha|$ and $|\beta|$ by $C=1+\vartheta(0;iT/2\pi)$, thus
\begin{align*}
0\leq S_{N,2}\leq & C\sum_{\alpha}q^{\frac{1}{2}|\alpha|}\sum_{|\beta|>N^{\gamma}} q^{\frac{1}{2}|\gamma|}\\
= & C\sum_{\alpha}q^{\frac{1}{2}|\alpha|} \sum_{k>N^{\gamma}} p(k) q^{\frac{k}{2}}\to 0,~\text{as } N\to\infty.
\end{align*}
This concludes the proof in the same way as in the special unitary case.
\end{proof}

\section{Non-orientable surfaces}\label{sec:norient}

We now turn to the study of non-orientable surfaces. Let us recall that any such surface can be constructed as the connected sum of $g$ projective planes. In order to estimate the large $N$ asymptotics of its associated partition function, we need to compute the Frobenius--Schur indicator associated to any highest weight.

\subsection{Frobenius--Schur indicator of a highest weight of $\U(N)$ or $\SU(N)$}

Let $(\rho,V)$ a complex finite-dimensional representation of a compact group $G$ of character $\chi_V$. The Frobenius--Schur indicator
\[
\iota_{\chi_V}=\int_G\chi_V(g^2)\mathrm{d} g
\]
appears in particular in the study of symmetric and alternating parts of the tensor product representation
\[
V\otimes V = \mathrm{Sym}^2 V\oplus {\bigwedge}^2 V.
\]
Indeed, straightforward computations involving the canonical bases of $V\otimes V$, $\text{Sym}^2 V$ and ${\bigwedge}^2 V$ yield
\begin{equation}\label{eq:fs1}
\chi_V(g^2) = \chi_V(g)^2 - 2\chi_{{\bigwedge}^2 V}(g)
\end{equation}
and
\begin{equation}\label{eq:fs2}
\chi_V(g^2) = 2\chi_{\text{Sym}^2 V}(g)-\chi_V(g)^2.
\end{equation}
Furthermore, $\rho$ is said to be:
\begin{enumerate}
\item \emph{Real} if it exists a symmetric $G$-invariant nondegenerate bilinear form;
\item \emph{Quaternionic} if it exists a skew-symmetric $G$-invariant nondegenerate bilinear form;
\item \emph{Complex} if there is no $G$-invariant nondegenerate bilinear form.
\end{enumerate}
The value of $\iota_{\chi_V}$ is actually based on this classification, as stated by the following Proposition.

\begin{proposition}\label{prop:fbi}
Let $(\rho,V)$ be a complex finite-dimensional representation of a compact group $G$, with character $\chi$. Its Frobenius--Schur indicator satisfies the following equation:
\begin{equation}
\iota_{\chi_V} = \left\lbrace \begin{array}{ll}
1 & \text{ if } \langle \chi_\text{triv},\chi_{\text{Sym}^2 V}\rangle = 1,\ i.e. \ \rho \text{ is real;}\\
-1 & \text{ if } \langle \chi_\text{triv},\chi_{\bigwedge^2 V}\rangle = 1,\ i.e. \ \rho \text{ is quaternionic;}\\
0 & \text{ otherwise}, \ i.e. \ \rho \text{ is complex.}
\end{array}\right.
\end{equation}
\end{proposition}

\begin{proof}
If we sum up Equations \eqref{eq:fs1} and \eqref{eq:fs2}, we have
\[
2\iota_{\chi_V} = \int_G 2\chi_{\text{Sym}^2 V}(g)\mathrm{d}g-2\int_G\chi_{\bigwedge^2 V}(g)\mathrm{d}g = 2\left(\langle \chi_\text{triv},\chi_{\text{Sym}^2 V}\rangle-\langle\chi_\text{triv},\chi_{\bigwedge^2 V}\rangle\right).
\]
Then, as $\text{Sym}^2 V$ and $\bigwedge^2 V$ are in direct sum, it appears that there are 3 cases:
\begin{itemize}
\item $\langle \chi_\text{triv},\chi_{\text{Sym}^2 V}\rangle = 1$ and $\langle \chi_\text{triv},\chi_{\bigwedge^2 V}\rangle = 0$, which gives $\iota_{\chi} = 1$;
\item $\langle \chi_\text{triv},\chi_{\text{Sym}^2 V}\rangle = 0$ and $\langle \chi_\text{triv},\chi_{\bigwedge^2 V}\rangle = 1$, which gives $\iota_{\chi} = -1$;
\item $\langle \chi_\text{triv},\chi_{\text{Sym}^2 V}\rangle = 0$ and $\langle \chi_\text{triv},\chi_{\bigwedge^2 V}\rangle = 0$, which gives $\iota_{\chi} = 0$.
\end{itemize}
\end{proof}

As a consequence of this result, computing the Frobenius--Schur indicator of an irreducible representation of $\U(N)$ or $\SU(N)$ can be done by determining whether the representation is real, complex or quaternionic. The following theorem gives a classification depending on the highest weight.

\begin{theorem}\label{cor:fh}
Let $\lambda\in\widehat{\SU}(N)$ be a highest weight and $n\in\Z$ be an integer.
\begin{enumerate}
\item  If $N=2M+1$ is odd and $(\alpha,\beta)\in\widehat{\SU}(M+1)^2$ such that $\lambda=\lambda(\alpha,\beta)$, then an irreducible representation of $\SU(N)$ with highest weight $\lambda$ is complex iff $\alpha\neq\beta$. Moreover, an irreducible representation of $\U(N)$ with highest weight $\lambda+n$ is real if $\alpha=\beta$ and $n=-\alpha_1$, otherwise it is complex.
\item If $N=2M$ is even, $(\alpha,\beta)\in\widehat{\SU}(M)\times\widehat{\SU}(M+1)$ such that $\lambda=\lambda(\alpha,\beta)$, then set $\tilde{\beta}=(\beta_1-\beta_M,\ldots,\beta_{M-1}-\beta_M,0)\in\widehat{\SU}(M)$, and an irreducible representation of $\SU(N)$ with highest weight $\lambda$ is complex iff $\alpha\neq\tilde{\beta}$. Moreover, an irreducible representation of $\U(N)$ with highest weight $\lambda+n$ is real if $\alpha=\tilde{\beta}$ and $n=-\alpha_1$, otherwise it is complex.
\item If $N$ is large enough and $\lambda\in\Lambda_{N,1}$ is an almost flat highest weight, then there is no quaternionic irreducible representation of $\SU(N)$ with highest weight $\lambda$.
\end{enumerate}
\end{theorem}

Note that when $\lambda=\lambda(\alpha,\beta)$ and $\alpha=\beta$ or $\alpha=\tilde{\beta}$ (depending on the parity of $N$), the integer $\lambda_1=\alpha_1+\beta_1=2\alpha_1$ is always even, so that the condition $n=\tfrac{\lambda_1}{2}$ makes sense. The main point of this theorem is that highest weights that are not symmetric are complex and therefore do not contribute to the non-orientable partition function because their Frobenius--Schur indicator vanishes. We can also notice that quaternionic representations of $\SU(N)$ with almost flat highest weight do not appear in the large $N$ scale, and that the partition function becomes a sum of nonnegative terms.

The proof of Theorem \ref{cor:fh} will rely on two propositions.

\begin{proposition}[\cite{FH}, Prop.26.24]\label{prop:fh}
Let $\lambda=(\lambda_1\geq\cdots\geq\lambda_{N}=0)$ be a highest weight of $\SU(N)$. Let $m_i=\lambda_i-\lambda_{i+1}\in\N$ for every $i\in\{1,\ldots,N-1\}$. An irreducible representation of $\SU(N)$ with highest weight $\lambda$ is:
\begin{itemize}
\item Complex if there exists $i$ such that $m_i\neq m_{N-i}$;
\item Real if for all $i$ $m_i=m_{N-i}$ and one of the following cases is satisfied:
\begin{itemize}
\item $N$ is odd;
\item $N=4k$ for a given $k\in\N$;
\item $N=4k+2$ for a given $k\in\N$ and $m_{2k+1}$ is even;
\end{itemize}
\item Quaternionic if for all $i$ $m_i=m_{N-i}$, $N=4k+2$ for a given $k\in\N$ and $m_{2k+1}$ is odd.
\end{itemize}
\end{proposition}

\begin{proposition}[\cite{BtD},§5.2]\label{prop:fh2}
Let $(\pi,V)$ be an irreducible representation of $\U(N)$ of highest weight $\lambda$. If it is self-conjugate, that is, $\lambda_i=-\lambda_{N+1-i}$ for all $1\leq i\leq N$, then it is real, otherwise it is complex.
\end{proposition}

\begin{proof}[Proof of Theorem \ref{cor:fh}]
$(i)$ and $(ii)$ are direct consequences of Prop. \ref{prop:fh} and \ref{prop:fh2}. $(iii)$ follows from the fact that for $N=4k+2$ with $k$ large enough, if $\lambda=\lambda(\alpha,\beta)$ is almost flat, then there is no `jump' between $\lambda_{2k+1}$ and $\lambda_{2k+2}$, thus $m_{2k+1}=0$ is always even.
\end{proof}

\subsection{Non-orientable surfaces of genus $g\geq 3$}
\subsubsection*{The special unitary case}

The proof of Theorem \ref{thm:mainno}.(i) will be based on the same reasoning as for orientable surfaces of genus $g\geq 2$, that is, using Proposition \ref{prop:upperbound} to show that the contribution of all other highest weights than $(0,\ldots,0)$ vanish in the large $N$ limit. However, the case of non-orientable surfaces with $g=3$ will need a finer control, as we will see later. In particular, for even values of $N$ and $g=3$ the following inequality is needed.

\begin{proposition}\label{prop:even3}
Let $N=2M$ be an integer, $\alpha\in\widehat{\SU}(M)$ and $\beta\in\widehat{\SU}(M+1)$ be two highest weights. We define $\lambda(\alpha,\beta)$ as in Section $2.2$, and $\tilde{\beta}=(\beta_1-\beta_M,\ldots,\beta_{M-1}-\beta_M,0)\in\widehat{\SU}(N)$. Then,
\[d_{\lambda(\alpha,\beta)}\geq \left(1+\frac{\beta_M}{M}\right)^M d_\alpha d_{\tilde{\beta}}.\]
\end{proposition}
\begin{proof}
Using Equation \eqref{eq:dim} and the fact that
\[\lambda(\alpha,\beta)=(\alpha_1+\beta_1,\ldots,\alpha_{M-1}+\beta_1,\beta_1,\beta_M-\beta_1,\ldots,\beta_2-\beta_1,0),\]
it is clear that $d_{\lambda(\alpha,\beta)}\geq d_\alpha d_\beta$. Moreover,
\begin{align*}
d_\beta = & \prod_{1\leq i<j\leq M+1} \left(1+\frac{\beta_i-\beta_j}{j-i}\right)\\
= & \prod_{i=1}^M\left(\frac{\beta_i}{M+1-i}\right)d_{\tilde{\beta}}\\
\geq & \left(\frac{\beta_M}{M}\right)^M d_{\tilde{\beta}}.
\end{align*}
Combining both inequalities gives the expected result.
\end{proof}

\begin{proof}[Proof of Theorem \ref{thm:mainno}.(i)]
The highest weight $(0,\ldots,0)$ is associated to the trivial representation, which is real by Proposition \ref{prop:fbi} and has dimension $1$ and Casimir number $0$. We can then rewrite
\[Z_N^{'-}(g,T) = 1 + \sum_{\substack{\lambda\in\widehat{\SU}(N)\\ \lambda\neq(0,\ldots,0)}} q^{c'_2(\lambda)}d_\lambda^{2-g} (\iota_\lambda)^g,\]
and the remaining sum can be bounded as follows:
\[\bigg\vert\sum_{\substack{\lambda\in\widehat{\SU}(N)\\ \lambda\neq(0,\ldots,0)}} q^{c'_2(\lambda)}d_\lambda^{2-g} \iota_\lambda\bigg\vert\leq \sum_{\substack{\lambda\in\widehat{\SU}(N)\\ \lambda\neq(0,\ldots,0)}} q^{c'_2(\lambda)}d_\lambda^{2-g}.\]
If $g\geq 4$, then the right-hand side has been proved to converge to $0$ as $N\to\infty$ in the proof of Theorem \ref{thm:main}, hence the result follows.

Now, if $g=3$, we need to refine the analysis in order to get the convergence. From Theorem \ref{cor:fh}, it appears that $\lambda\in\SU(N)$ contributes to the partition function iff it is symmetric. The case $N=2M+1$ is easier to prove, so we start with it. As $\iota_\lambda=0$ if $\lambda$ is associated with a complex representation, we have
\begin{align*}
Z_N^{'-}(3,T) = & 1 + \sum_{\substack{\lambda\in\widehat{\SU}(N)\\ \lambda\neq(0,\ldots,0)\\ \lambda\text{ is symmetric}}} q^{c'_2(\lambda)}d_\lambda^{-1} (\iota_\lambda)^3,
\end{align*}
which means that
\begin{align*}
|Z_N^{'-}(3,T)-1| = & \bigg\vert\sum_{\substack{\alpha\in\widehat{\SU}(M+1)\\ \alpha\neq(0,\ldots,0)}} q^{c'_2(\lambda(\alpha,\alpha))}d_{\lambda(\alpha,\alpha)}^{-1} (\iota_{\lambda(\alpha,\alpha)})^3\bigg\vert\\
\leq & \sum_{\substack{\alpha\in\widehat{\SU}(M+1)\\ \alpha\neq(0,\ldots,0)}} q^{c'_2(\lambda(\alpha,\alpha))}d_\alpha^{-2}\\
\leq  & \zeta_{\mathfrak{su}(M)}(2).
\end{align*}
Then, letting $M$ tend to infinity and using Proposition \ref{prop:upperbound}, we have indeed
\[\lim_{M\to\infty} Z_{2M+1}^{'-}(3,T) = 1.\]

Now consider $N=2M$. Let $\tilde{\beta}=(\beta_1-\beta_M,\ldots,\beta_{M-1}-\beta_M,0)$. Theorem \ref{cor:fh} states that $\lambda(\alpha,\beta)$ contributes to the partition function iff $\alpha=\tilde{\beta}$. It implies:
\begin{align*}
|Z_N^{'-}(3,T)-1| = & \bigg\vert\sum_{\substack{(\alpha,\beta)\in\widehat{\SU}(M)\times\widehat{\SU}(M+1)\\ \alpha=\tilde{\beta}}} q^{c'_2(\lambda(\alpha,\beta))}d_{\lambda(\alpha,\beta)}^{-1} (\iota_{\lambda(\alpha,\alpha)})^3\bigg\vert.
\end{align*}
We can then apply Proposition \ref{prop:even3} to get
\begin{align*}
|Z_N^{'-}(3,T)-1| \leq & \sum_{\substack{\alpha\in\widehat{\SU}(M+1)\\ \alpha\neq(0,\ldots,0)}}\sum_{n\in\N}\left(1+\frac{n}{M}\right)^{-M} d_\alpha^{-2}\\
= & \sum_{n\in\N}\left(1+\frac{n}{M}\right)^{-M} \sum_{\substack{\alpha\in\widehat{\SU}(M+1)\\ \alpha\neq(0,\ldots,0)}}d_\alpha^{-2}.
\end{align*}
The first sum is bounded because $\left(1+\frac{n}{M}\right)^{-M}\leq e^{-n}$ for any $n,M$, and the second one converges, following the same argument as in the case $N=2M+1$. We finally get
\[\lim_{M\to\infty} Z_{2M}^{'-}(3,T) = 1.\]
\end{proof}

\subsubsection*{The unitary case}

As for the special unitary case, the proof of the unitary case for non-orientable surfaces of genus $g\geq 3$ is similar to the one of orientable surfaces of genus $g\geq 2$. Indeed, the point is to show that only constant highest weights contribute to the large $N$ limit.

\begin{proof}[Proof of Theorem \ref{thm:mainno}.(i) in the unitary case]Let us consider $g\geq 3$ and $T>0$. The only constant highest weight of $\U(N)$ corresponding to a non-complex irreducible representation is $(0,\ldots,0)$, and has Frobenius--Schur indicator equal to 1. We can then split the partition function $Z_{N}^-(g,T)$ into two parts:
\[Z_{N}^-(g,T)=1 +\sum_{\substack{\lambda\in \widehat\SU(N)\\ \lambda\neq (0,\ldots,0)}}\sum_{n\in \Z} e^{-\frac{T}{2}c_{2}(\lambda+n)}d_{\lambda+n}^{2-g}\iota_{\lambda+n}^g.\]
Now, given $\lambda\in\widehat{\SU}(N)$ and $n\in\Z$, we know that a necessary and sufficient condition for $\iota_{\lambda+n}$ to be nonzero is that $\lambda_1=-2n$, therefore we have
\[
\left\vert Z_{N}^-(g,T)-1\right\vert=\bigg\vert\sum_{\substack{\lambda\in \widehat\SU(N)\\ \lambda\neq (0,\ldots,0)\\ \lambda_1 \text{ is even}}}q^{c_2(\lambda-\tfrac{\lambda_1}{2})}d_{\lambda}^{2-g}\bigg\vert\leq \sum_{\substack{\lambda\in \widehat\SU(N)\\ \lambda\neq (0,\ldots,0)}}d_{\lambda}^{2-g}.
\]
We are now in the same setting as in the special unitary case, and the convergence follows from the same arguments.
\end{proof}

\subsection{The Klein bottle}

The Klein bottle is the non-orientable equivalent to the torus, as we will see, in the sense that the dimension of the irreducible representations do not appear in the formula of the partition function. Hence, the proof of Theorem \ref{thm:mainno}.(ii) is using almost flat highest weights as well.
\subsubsection*{The special unitary case}
\begin{proof}[Proof of Theorem \ref{thm:mainno}.(ii) in the special unitary case]
Let $\gamma\in(0,\frac{1}{2})$, and the subsets $(\Lambda_{N,i})_{1\leq i\leq 4}$ of $\widehat{\SU}(N)$ as in the case of the torus. We define, for $1\leq i\leq 4$,
\[S'_{N,i}=\sum_{\lambda\in\Lambda_{N,i}}\iota_{\lambda(\alpha,\beta)}^2 q^{c'_2(\lambda)}=\sum_{\lambda\in\Lambda_{N,i}} q^{c'_2(\lambda)},\]
and we obtain the following decomposition:
\[Z_N^{'-}(1,T)=S'_{N,1}+S'_{N,2}+S'_{N,3}+S'_{N,4}.\]

Let $\lambda(\alpha,\beta)$ be an element of $\Lambda_{N,1}$. We will discuss the case when $N$ is even and the case when it is odd, and show that the subsequences $(Z_{2M}^{'-})$ and $(Z_{2M+1}^{'-})$ both converge to the same limit.
\begin{itemize}
\item If $N=2M+1$, from Theorem \ref{cor:fh} we know that $\iota_{\lambda(\alpha,\beta)}^2=1$ if $\alpha=\beta$, and 0 otherwise. If this is the case, we can simplify Equation \eqref{eq16} into
\begin{equation}\label{eq:c2lambdaalpha}
c'_2(\lambda(\alpha,\alpha))=2|\alpha|+\frac{4K(\alpha)}{N},
\end{equation}

for any $\alpha$ of length $r$ and $N\geq2r$. Let us recall the estimation
\[|2K(\alpha)|\leq |\alpha|(|\alpha|-1),\]
which leads for $\lambda(\alpha,\alpha)\in\Lambda_{N,1}$ to
\[|c'_2(\lambda(\alpha,\alpha))-2|\alpha||\leq 4N^{2\gamma-1}.\]
We then get the estimate
\begin{equation}\label{zkleinoddsu}
q^{4N^{2\gamma-1}} \sum_{|\alpha|\leq N^\gamma} q^{2|\alpha|} \leq S'_{N,1} \leq q^{-4N^{2\gamma-1}} \sum_{|\alpha|\leq N^\gamma} q^{2|\alpha|},
\end{equation}
and both bounds converge to the expected quantity $\prod_{m=1}^\infty \frac{1}{1-q^{2m}}$.
\item If $N=2M$, let $\tilde{\beta}=(\beta_1-\beta_M,\ldots,\beta_{M-1}-\beta_M,0)$ as in the $g\geq 3$ case. We know from Theorem \ref{cor:fh} that $\iota_{\lambda(\alpha,\beta)}=1$ if $\alpha=\tilde{\beta}$ and 0 otherwise, so we have
\[
\beta_M=\frac{2(|\beta|-|\alpha|)}{N}.
\]
The condition $|\alpha|,|\beta|\leq N^\gamma$ is then equivalent to
\[
\left\lbrace \begin{array}{ccc}
|\alpha| & \leq & N^\gamma\\
\beta_M & \leq & \tfrac12 N^{\gamma-1}-\frac{|\alpha|}{2N}
\end{array}\right..
\]
Furthermore, Equation \eqref{eq:encadre} becomes
\[
2|\alpha|+M\beta_M-4N^{2\gamma-1} \leq c'_2(\lambda(\alpha,\beta)) \leq 2|\alpha|+M\beta_M + 4N^{2\gamma-1}+4N^{2\gamma-2}.
\]
We obtain that
\begin{equation}\label{zkleinevensu}
S'_{N,1} \leq 
q^{-4N^{2\gamma-1}}\sum_{\vert\alpha\vert\leq N^\gamma}\Bigg(\sum_{0\leq n\leq \tfrac12 N^{\gamma-1}-\frac{|\alpha|}{2N}}q^{Mn}\Bigg)q^{2\vert\alpha\vert},
\end{equation}
\begin{equation}\label{zkleinevensu2}
S'_{N,1}\geq q^{4N^{2\gamma-1}+4N^{2\gamma-2}} \sum_{\vert\alpha\vert\leq N^\gamma}\Bigg(\sum_{0\leq n\leq \tfrac12 N^{\gamma-1}-\frac{|\alpha|}{2N}}q^{Mn}\Bigg)q^{2\vert\alpha\vert}.
\end{equation}
The sums over $n$ are bounded between 1 and $\sum_{n\in\N}q^{Mn}$ which is bounded because $|q^M|<1$ and converges to 1 as $N$ tends to infinity (by dominated convergence). It finally appears that both bounds of \eqref{zkleinevensu} and \eqref{zkleinevensu2} converge to $\prod_{m=1}^\infty \frac{1}{1-q^{2m}}$.
\end{itemize}

By similar arguments as the ones used in the case of the torus, we can prove that $S'_{N,2}$, $S'_{N,3}$ and $S'_{N,4}$ all converge to 0 as the remainders of convergent series. This concludes the proof.
\end{proof}

\subsubsection*{The unitary case}

\begin{proof}[Proof of Theorem \ref{thm:mainno}.(ii) in the unitary case]
Let us start from the definition of $Z_N^-(2,T)$. We have
\[
Z_N^-(2,T)=\sum_{\lambda(\alpha,\beta)\in\widehat{\SU}(N)}\sum_{n\in\Z} q^{c_2(\lambda(\alpha,\beta)+n)}(\iota_{\lambda(\alpha,\beta)+n})^2.
\]
We know from Corollary \ref{cor:fh} that $\iota_{\lambda(\alpha,\beta)+n}=1$ if $\lambda(\alpha,\beta)$ is symmetric and $n=-\tfrac{\lambda}{2}=-\alpha_1$, and 0 otherwise. We can then simplify the formula into
\[
Z_N^-(2,T)=\sum_{\substack{\lambda(\alpha,\beta)\in\widehat{\SU}(N)\\ \lambda \text{ is symmetric}}}q^{c_2(\lambda(\alpha,\beta)-\alpha_1)}.
\]
As in the special unitary case, we will distinguish between the odd and even values of $N$, and prove that
\[
\lim_{M\to\infty} Z_{2M}^{'-}=\lim_{M\to\infty} Z_{2M+1}^{'-}=\frac{1}{\phi(q)^2},
\]
which implies the convergence of $(Z_N^-)$.
\begin{itemize}
\item If $N=2M+1$, the symmetry condition is equivalent to $\alpha=\beta$, and in particular ${\lambda(\alpha,\beta)-\beta_1}=\lambda(\alpha,\beta,0)$. Its Casimir number is given in Equation \eqref{eq15}:
\[
c_2(\lambda(\alpha,\beta,0))=|\alpha|+|\beta|+\frac{2}{N}(K(\alpha)+K(\beta)).
\]
Comparing with \eqref{eq16} we remark that $c_2(\lambda(\alpha,\beta,0))=c'_2(\lambda(\alpha,\beta))$. Then,
\[
Z_{2M}^-(2,T)=Z_{2M}^{'-}(2,T)
\]
and we can conclude from the special unitary case.
\item If $N=2M$, let $\tilde{\beta}=(\beta_1-\beta_M,\ldots,\beta_{M-1}-\beta_M,0)$ as in the $g\geq 3$ case, then the symmetry condition is equivalent to the fact that $\alpha=\tilde{\beta}$ and we have
\[
\beta_M=\frac{2(|\beta|-|\alpha|)}{N}.
\]
Let $\gamma\in(0,\frac{1}{2})$, and the subsets $(\Lambda_{N,i})_{1\leq i\leq 4}$ of $\widehat{\SU}(N)$ as usual. We define, for $1\leq i\leq 4$,
\[S_{N,i}=\sum_{\substack{\lambda(\alpha,\beta)\in\Lambda_{N,i}\\ \alpha=\tilde{\beta}}}q^{c_2(\lambda(\alpha,\beta,0))},\]
and we obtain the following decomposition:
\[Z_N^-(1,T)=S_{N,1}+S_{N,2}+S_{N,3}+S_{N,4}.\]
The condition $|\alpha|,|\beta|\leq N^\gamma$ is then equivalent to
\[
\left\lbrace \begin{array}{ccc}
|\alpha| & \leq & N^\gamma\\
\beta_M & \leq & \tfrac12 N^{\gamma-1}-\frac{|\alpha|}{2N}
\end{array}\right..
\]
From \eqref{eq15} and \eqref{eq16} we have
\[
c_2(\lambda(\alpha,\beta,0))=c'_2(\lambda(\alpha,\beta))-\frac{(|\alpha|-|\beta|)^2}{N^2}=c'_2(\lambda(\alpha,\beta))-4\beta_M^2.
\]
We can combine all these estimations with \eqref{eq:encadre} to obtain
\[
2|\alpha|+M\beta_M-4N^{2\gamma-1}-4\beta_M^2 \leq c_2(\lambda(\alpha,\beta,0)) \leq 2|\alpha|+M\beta_M + 4N^{2\gamma-1}+4N^{2\gamma-2}-4\beta_M^2.
\]
Recall that for $\lambda(\alpha,\beta)\in\Lambda_{N,1}$ we have $\beta_M\leq 2\frac{N^\gamma-|\alpha|}{N}\leq 2N^{\gamma-1}$, which yields
\[
\beta_M^2\leq  4N^{2\gamma-2}.
\]
We obtain therefore new bounds for $c_2(\lambda(\alpha,\beta,0))$:
\[
2|\alpha|+M\beta_M-4N^{2\gamma-1}-4N^{2\gamma-2} \leq c_2(\lambda(\alpha,\beta,0)) \leq 2|\alpha|+M\beta_M + 4N^{2\gamma-1}+4N^{2\gamma-2}.
\]

We obtain similar bounds for $S_{N,1}$ as we did in \eqref{zkleinevensu} and \eqref{zkleinevensu2}, and we recover the same limit using the same arguments. We can also use the same routine as we did for the torus, to prove the convergence of $S_{N,2}$, $S_{N,3}$ and $S_{N,4}$ to 0 as the remainders of convergent series.
\end{itemize}
\end{proof}

\section*{Acknowledgements}

The author would like to thank his PhD advisor, Thierry L\'evy, for introducing him to Yang--Mills theory but also for all his help in writing this article, as well as the anonymous referee who pointed out a mistake in the preliminary version of Theorem \ref{thm:mainno} and its proof for $\U(N)$. He would also like to thank Antoine Dahlqvist for several helpful discussions about the partition function on the torus.

\bibliographystyle{plain}
\bibliography{PFbib}

\end{document}